\documentclass[12pt]{amsart}
\usepackage{a4wide}
\usepackage[dvips]{epsfig}
\usepackage{amscd}
\usepackage{amssymb}
\usepackage{amsthm}
\usepackage{amsmath}
\usepackage{latexsym}
\usepackage[linktocpage,
                  menucolor=blue,
                  linkcolor=blue,
                  citecolor=red,
                  colorlinks=true]{hyperref}

\theoremstyle{plain}
\newtheorem{theorem}{Theorem}[section]
\theoremstyle{plain}
\newtheorem{lemma}[theorem]{Lemma}
\newtheorem{proposition}[theorem]{Proposition}

\theoremstyle{definition}
\newtheorem{definition}{Definition}[section]

{%
\setcounter{enumi}{0}

\begin{enumerate}}%
{\end{enumerate} }
{%
\setcounter{enumi}{0}

\begin{enumerate}}%
{\end{enumerate} }

\numberwithin{equation}{section}
 \allowdisplaybreaks

\title[Optimal portfolios, consumption and  insurance]
 {Optimal investment-consumption and life insurance with capital constraints}

\date{\today}

\begin{document}

\author{Calisto Guambe}
\address{Department of Mathematics and Applied Mathematics, University of Pretoria, 0002, South Africa }
\address{Department of Mathematics and Computer Sciences, Eduardo Mondlane University, 257, Mozambique}

\email{calistoguambe@yahoo.com.br}

\author{ Rodwell Kufakunesu }

\address{Department of Mathematics and Applied Mathematics
, University of Pretoria, 0002, South Africa}

\email{rodwell.kufakunesu@up.ac.za}



\begin{abstract}
The aim of this paper is to solve an optimal investment, consumption and life insurance problem when the investor is restricted to capital guarantee. We consider an incomplete market described by a jump-diffusion model with stochastic volatility. Using the martingale approach, we prove the existence of the optimal strategy and the optimal martingale measure and we obtain the explicit solutions for the power utility functions.
\end{abstract}

\maketitle
\section{Introduction}
Optimal consumption-investment problem by Merton \cite{Merton} ushered a lot of extensions. In 1975, Richard \cite{Richard} extended for the first time this problem to include life insurance decisions. Other references include (Huang {\it et al.} \cite{Huang}, Pliska and Ye \cite{Pliska}, Liang and Guo \cite{liang}). Recently, Kronborg and Steffensen \cite{steffensen}
extended this problem to include capital constraints, previously introduced by Tepl\'a \cite{Tepla} and El Karoui {\it et. al.} \cite{elkaroui}. Most of the references mentioned above solved the problem under a diffusion framework.

However, as was pointed out by Merton and many empirical data, the analysis of the price evolution reveals some sudden and rare breaks (jumps) caused by external information flow. These behaviours constitute a very real concern of most investors and they can be modeled by a Poisson process, which has jumps occurring at rare and unpredictable time. For detailed information see e.g., Jeanblanc-Picque and Pontier \cite{pontier}, Runggaldier \cite{runggaldier}, Daglish \cite{daglish}, Oksendal and Sulem \cite{Oksendal}, Hanson \cite{hanson} and references therein.

In this paper, we consider a jump-diffusion problem with stochastic volatility as in Mnif \cite{mnif}. In his paper, Mnif \cite{mnif} solved the portfolio optimization problem using the dynamic programming approach. Applying this technique in a jump-diffusion model, the Hamilton-Jacobi-Bellman (HJB) equation associated to the problem is nonlinear, which in general the explicit solution is not provided. To prove the existence of a smooth solution, he reduced the nonlinearity of the HJB equation to a semi-linear equation under certain conditions. We use a martingale approach developed by Karatzas {\it et al.} \cite{karatzas87} and Karatzas and Shreve \cite{karatzas98} in a diffusion process to solve the unrestricted problem. Considering a jump-diffusion model, a market is incomplete and consequently we have many martingale measures. We obtain the optimal investment, consumption and life insurance strategy by the convex optimization method. This method allow us to characterize the optimal martingale measure for the utility functions of the power type. In the literature, this method has also been applied by Castaneda-Leyva and Hern\'andez-Hern\'andez \cite{hernandez} in an optimal investment-consumption problem considering a stochastic volatility model described by diffusion processes. Similar works include (Liang and Guo \cite{liang}, Michelbrink and Le \cite{michelbrink} and references therein).

The optimal solution to the restricted problem is derived from the unrestricted optimal solution, applying the option based portfolio insurance (OBPI) method developed by El Karoui {\it et al.} \cite{elkaroui}. The OBPI method consists in taking a certain part of capital and invest in the optimal portfolio of the unconstrained problem and the remaining part insures the position with American put. We prove the admissibility and the optimality of the strategy.

The structure of this paper is organized as follows. In Section 2, we introduce the model and problem formulation of the Financial and the Insurance markets. Section 3, we solve the unconstrained problem. In Section 4, we solve the constrained problem and prove the admissibility of our strategy. Finally, in Section 5 we give a conclusion.

\section{The Financial Model}

We consider two dimensional Brownian motion $W=\{W_1(t);W_2(t), \ 0\leq t\leq T\}$ associated to the complete filtered probability space $(\Omega^W,\mathcal{F}^W,\{\mathcal{F}_t^W\},\mathbb{P}^W)$ such that $\{W_1(t),W_2(t)\}$ are correlated with the correlation coefficient $|\varrho|<1$, that is, $dW_1(t)\cdot dW_2(t)=\varrho dt$.  Moreover, we consider a Poisson process $N=\{N(t),\mathcal{F}^N(t),0\leq t\leq T\}$ associated to the complete filtered probability space $(\Omega^N,\mathcal{F}^N,\{\mathcal{F}^N_t\},\mathbb{P}^N)$ with intensity $\lambda(t)$ and a $\mathbb{P}^N$-martingale compensated poisson process
\begin{equation*}
    \tilde{N}(t):=N(t)-\int_0^t\lambda(t)dt\,.
\end{equation*}
We assume that the intensity $\lambda(t)$ is Lebesgue integrable on $[0,T]$.

Consider the product space:
$$(\Omega,
\mathcal{F},\{\mathcal{F}_t\}_{0\leq t\leq T},\mathbb{P}):=(\Omega^W\times\Omega^N,\mathcal{F}^W\otimes\mathcal{F}^N, \{\mathcal{F}^W_t\otimes\mathcal{F}^N_t\},\mathbb{P}^W\otimes\mathbb{P}^N)\,,
$$ where
$\{\mathcal{F}_t\}_{t\in[0,T]}$ is a filtration satisfying the usual
conditions (Protter \cite{Protter}).  On this space, we assume that $W$ and $N$ are independent processes.

The financial market consists of a risk-free asset $B:=(B(t)_{t\in[0,T]})$, a non-tradable index $Z:=(Z(t)_{t\in[0,T]})$ which can be thought as an external economic factor, such as a temperature, a loss index or a volatility driving factor and a risky asset $S:=(S(t))_{t\in[0,T]}$ correlated with $Z(t)$. This market is defined by the following jump-diffusion model:
\begin{eqnarray}\label{risk-free}
  dB(t) &=& r(t)B(t)dt,\quad B(0)=1,    \\\label{nontradable}
  dZ(t)&=& \eta(Z(t))dt+ dW_1(t),  \\ \label{risky}
  dS(t) &=&
S(t)\Bigl[\alpha(t,Z(t))dt+\beta(t,Z(t))dW_1(t)+\sigma(t,Z(t))dW_2 \\  \nonumber
&& \ \ \ \ \ \ \ \ +\gamma(t,Z(t))dN(t)\Bigl]\,, \ \ S(0)=s>0\,,
\end{eqnarray}
 where $r(t)$ is the risk-free interest rate, $\alpha(t,z),\beta(t,z),\sigma(t,z)$ and $\gamma(t,z)>-1$ denote the mean rate of return, volatility rates and the dispersion rates respectively. With the latter condition and the continuity of $Z$, we guarantee that \eqref{risky} is well defined.

We assume that the market parameters satisfy the following conditions:\\
\begin{itemize}
\item[${\bf(A_1)}$] The functions  $r:[0,T]\rightarrow\mathbb{R}$;  $\alpha,\beta,\sigma,\gamma:[0,T]\times \mathbb{R}\rightarrow\mathbb{R}$ belong to $C^{1,1}([0,T]\times\mathbb{R})$ with bounded derivatives.  Moreover, for all $(t,z)\in[0,T]\times\mathbb{R}$ and $K>0$,
 \begin{eqnarray*}
   |\alpha(t,z)|+|\beta(t,z)| &\leq& K  \\
   |\sigma(t,z)| &\leq& K(1+|z|)\,.
 \end{eqnarray*}
\end{itemize}
To ensure the existence and uniqueness of the solution to \eqref{nontradable}, we assume a Lipschitz condition on the $\mathbb{R}$-valued function $\eta$:\\
\begin{itemize}
\item[${\bf(A_2)}$] There exists a positive constant $C$ such that $$ |\eta(y)-\eta(w)|\leq C|y-w|\,, \ \ \ \ \ \  y,w\in\mathbb{R}\,.$$
\end{itemize}
Under the above assumption, the solution to the stochastic differential equation (SDE) \eqref{nontradable} is given by
\begin{equation}\label{nontradabler}
Z(t)=z+\int_0^t\eta(Z(s))ds+\int_0^tdW_1(s)\,.
\end{equation}

 Let us consider a policyholder  whose lifetime is a
nonnegative random variable $\tau$ defined on the probability space $(\Omega,
\mathcal{F},\mathbb{P})$ and independent of the filtration $\mathcal{F}_t$.
Consider that $\tau$ has a probability density function $f(t)$ and the distribution function given by
$$
F(t):=\mathbb{P}(\tau<t)=\int_0^tf(s)ds\,.
$$
The probability that the life time $\tau>t$ is given by
$$
\bar{F}(t):=\mathbb{P}(\tau\geq t\mid\mathcal{F}_t)=1-F(t)\,.
$$
The instantaneous force of mortality $\mu(t)$ for the policyholder to be alive at time $t$ is defined by
\begin{eqnarray*}
  \mu(t) &:=& \lim_{\Delta t\rightarrow0}\frac{\mathbb{P}(t\leq\tau<t+\Delta t|\tau\geq t)}{\Delta t} \\
   &=& \lim_{\Delta t\rightarrow0}\frac{\mathbb{P}(t\leq\tau<t+\Delta t)}{\Delta t\mathbb{P}(\tau\geq t)} \\
   &=&  \frac{1}{\bar{F}(t)}\lim_{\Delta t\rightarrow0}\frac{F(t+\Delta t)-F(t)}{\Delta t} \\
   &=& \frac{f(t)}{\bar{F}(t)}=-\frac{d}{dt}(\ln(\bar{F}(t)))\,.
\end{eqnarray*}

 Then, the conditional survival probability of the policyholder is given by
 \begin{equation}\label{survival}
\bar{F}(t)=\mathbb{P}(\tau>t|\mathcal{F}_t)=\exp\left(-\int_0^t\mu(s)ds\right),
 \end{equation}
 and the conditional survival probability density of the death of the policyholder by
 \begin{equation}\label{death}
    f(t):=\mu(t)\exp\left(-\int_0^t\mu(s)ds\right).
 \end{equation}

Let $c(t)$ be the consumption rate of the policyholder, $\pi(t)$ be the amount of the policyholder's wealth invested in the risky asset $S$ and $p(t)$ be the sum insured to be paid out at time $t\in[0,T]$ for the life insurance upon the wage earner's death before time $T$. We assume that the strategy $(c(t),\pi(t),p(t))$ satisfy the following definition:
\begin{definition}
The consumption rate $c$ is measurable, $\mathcal{F}_t$-adapted process, nonnegative and
$$
\int_0^Tc(t)dt<\infty, \ \  \hbox{a.s.}
$$
The allocation process $\pi$ is an $\mathcal{F}_t$-predictable process with
$$\int_0^t\pi^2(t)dt<\infty, \ \ \hbox{a.s.}
$$
The insurance process $p$ is measurable, $\mathcal{F}_t$-adapted process, nonnegative and
$$ \int_0^Tp(t)dt<\infty, \ \ \hbox{a.s.}
$$

\end{definition}

Suppose that the policyholder receives a deterministic labor income of rate $\ell(t)\geq0$, $\forall t\in[0,\tau\wedge T]$ and that the shares are divisible and can be traded continuously. Furthermore, we assume that there are no transaction
costs, taxes or short-selling constraints in the trading, then after some calculations, the wealth process
$X(t)\,, \ t\in[0,\tau\wedge T]$ is defined by the following SDE:

 \begin{eqnarray}\label{wealth}
   dX(t) &=& \left[(r(t)+\mu(t))X(t)+\pi(t)(\alpha(t,Z(t))-r(t))+\ell(t)-c(t)-\mu(t)p(t)\right]dt \\ \nonumber
&+&\pi(t)\beta(t,Z(t)) dW_1(t)+\pi(t)\sigma(t,Z(t))dW_2(t)+\pi(t)\gamma(t,Z(t)) dN(t)\,, \\ \nonumber
   X(0) &=& x>0,
 \end{eqnarray}
where $Z(t)$ is given by \eqref{nontradabler} and $\tau\wedge
T:=\min\{\tau, T\}$.

The expression $\mu(t)(p(t)-X(t))dt$ from the wealth process \eqref{wealth}, corresponds to the risk premium rate to pay for the life insurance $p$ at time $t$. Notice that choosing $p>X$ corresponds to buying a life insurance and $p<X$ corresponds to selling a life insurance, that is, buying an annuity (Kronborg and Steffensen \cite{steffensen}).

From Definition 2.1 and the condition ${\bf (A_1)}$, we see that the wealth process (\ref{wealth}) is well defined and has a unique solution given by
\begin{eqnarray}\nonumber
  X(t) &=& x_0e^{\int_0^t(r(s)+\mu(s))ds}+\int_0^te^{\int_s^t(r(u)+\mu(u))du}\Bigl[ \pi(s)(\alpha(s,Z(s))-r(s)) +\ell(s)-c(s) \\ \nonumber
&& -\mu(s)p(s)\Bigl]ds + \int_0^t\pi(s)\beta(s,Z(s)) e^{\int_s^t(r(u)+\mu(u))du}dW_1(s) \\
   && + \int_0^t\pi(s)\sigma(s,Z(s)) e^{\int_s^t(r(u)+\mu(u))du}dW_2(s) + \int_0^t\pi(s)\gamma(s,Z(s)) e^{\int_s^t(r(u)+\mu(u))du}dN(s)\,.
\end{eqnarray}

We define a new probability measure $\mathbb{Q}$ equivalent to $\mathbb{P}$ in which  $S$ is a local martingale. The Radon-Nikodym derivative is given by:
\begin{eqnarray}\nonumber
    \Lambda(t)&:=&\exp\Bigl\{\int_0^t[(1-\psi(s))\lambda(s) -\frac{1}{2}\theta^2(s,Z(s),\psi(s))-\frac{1}{2}\nu^2(s,\psi(s))]ds  \\ \label{radom}
   && +\int_0^t\nu(s,Z(s),\psi(s))dW_1(s)+\int_0^t\theta(s,Z(s),\psi(s))dW_2(s)+\int_0^t\ln(\psi(s))dN(s)\Bigl\}.
\end{eqnarray}
By Girsanov's Theorem under $\mathbb{Q}$, we have that:

\begin{equation*}
    \left\{
       \begin{array}{lll}
         dW_1^{\mathbb{Q},\psi}(t)=dW_1(t)-\nu(t,Z(t),\psi(t))dt\,, \\
         dW_2^{\mathbb{Q},\psi}(t)=dW_2(t)-\theta(t,Z(t),\psi(t))dt\,, \\
         d\tilde{N}^{\mathbb{Q}}(t)=dN(t)-\psi(t)\lambda(t)dt\,
       \end{array}
     \right.
\end{equation*}
are Brownian motions and compensated Poisson random measure respectively, where (See Runggaldier \cite{runggaldier})
\begin{equation}\label{theta}
\nu(t,Z(s),\psi(t))= \frac{\beta(t,Z(t))}{\beta^2(t,Z(t))+\sigma^2(t,Z(t))}(r(t)-\alpha(t,Z(t))-\gamma(t,Z(t))\psi(t)\lambda(t))\,,
\end{equation}
\begin{equation}\label{psi}
  \theta(t,Z(t),\psi(t))= \frac{\sigma(t,Z(t))}{\beta^2(t,Z(t))+\sigma^2(t,Z(t))}(r(t)-\alpha(t,Z(t))-\gamma(t,Z(t))\psi(t)\lambda(t))\,,
\end{equation}
for any $\mathcal{F}_t$-adapted bounded measure $\psi>0$. We assume that $\beta^2(t,Z(t))+\sigma^2(t,Z(t))\neq0$. Thus we have infinitely many martingale measures and consequently incomplete market.

Note that, from the boundedness of the associated parameters, the predictable processes $\nu, \theta$, are bounded. Then, one can prove  that the stochastic exponential \eqref{radom} is a positive martingale (see Delong \cite{Delong}, Proposition 2.5.1.).\\

\noindent From (\ref{theta}) and (\ref{psi}), we have that:
\begin{eqnarray*}
    && [\pi(t)(\alpha(t,Z(t))-r(t))+\pi(t)\beta(t,Z(t))\nu(t,Z(t),\psi(t)) +\pi(t)\sigma(t,Z(t))\theta(t,Z(t),\psi(t)) \\
    && \ \  +\pi(t)\gamma(t,Z(t))\psi(t)\lambda(t)]\, =\, 0\,,
\end{eqnarray*}
then under $\mathbb{Q}$, the dynamics of the wealth process is given by
\begin{eqnarray}\nonumber
dX(t) &=& \left[(r(t)+\mu(t))X(t)+\ell(t)-c(t)-\mu(t)p(t)\right]dt+ \pi(t)\beta(t,Z(t)) dW_1^{\mathbb{Q},\psi}(t) \\ \label{wealthqq}
 && + \pi(t)\sigma(t,Z(t)) dW_2^{\mathbb{Q},\psi}(t) +\pi(t)\gamma(t,Z(t)) d\tilde{N}^{\mathbb{Q},\psi}(t)\,,
\end{eqnarray}
which gives the following representation:
\begin{eqnarray}\label{wealthq}
  X(t) &=& x_0e^{\int_0^t(r(s)+\mu(s))ds}+\int_0^te^{\int_s^t(r(u)+\mu(u))du}\Bigl[\ell(s)-c(s) -\mu(s)p(s)\Bigl]ds \\ \nonumber
 &&
 + \int_0^t\pi(s)\beta(s,Z(s)) e^{\int_s^t(r(u)+\mu(u))du}dW_1^{\mathbb{Q},\psi}(s) \\ \nonumber
 && + \int_0^t\pi(s)\sigma(s,Z(s)) e^{\int_s^t(r(u)+\mu(u))du}dW_2^{\mathbb{Q},\psi}(s) \\ \nonumber
 &&+ \int_0^t\pi(s)\gamma(s,Z(s)) e^{\int_s^t(r(u)+\mu(u))du}d\tilde{N}^{\mathbb{Q},\psi}(s)\,.
\end{eqnarray}

The following definition introduces the concept of admissible strategy.
\begin{definition}
Define the set of admissible strategies $\{\mathcal{A}\}$ as the consumption, investment and life insurance strategies for which the corresponding wealth process given by (\ref{wealthq}) is well defined and
\begin{equation}\label{admissible}
    X(t)+g(t)\geq0, \ \ \forall t\in[0,T],
\end{equation}
where $g$ is the time $t$ actuarial value of future labor income defined by
\begin{equation}\label{laborincome}
    g(t):=\mathbb{E}\left[\int_t^Te^{-\int_t^s(r(u)+\mu(u))du}\ell(s)ds\mid\mathcal{F}_t\right]\,.
\end{equation}
\end{definition}
Since
\begin{eqnarray} \label{martingalew1}
  \mathbb{E}^{\mathbb{Q},\psi}\left[\int_0^t\pi(s)\beta(s,Z(s)) e^{\int_s^t(r(u)+\mu(u))du}dW_1^{\mathbb{Q},\psi}(s)\right] &=& 0\,, \\ \label{martingalew2}
   \mathbb{E}^{\mathbb{Q},\psi}\left[\int_0^t\pi(s)\sigma(s,Z(s)) e^{\int_s^t(r(u)+\mu(u))du}dW_2^{\mathbb{Q},\psi}(s)\right] &=& 0\,, \\ \label{martingaleM}
  \mathbb{E}^{\mathbb{Q},\psi}\left[\int_0^t\pi(s)\gamma(s,Z(s)) e^{\int_s^t(r(u)+\mu(u))du}d\tilde{N}^{\mathbb{Q},\psi}(s)\right] &=& 0\,,
\end{eqnarray}
we see that the last three terms in (\ref{wealthq}) are $\mathbb{Q}$ local martingales and from (\ref{admissible}), a supermartingale (see e.g. Karatzas {\it et al.} \cite{karatzas}). Then, the strategy $(c,\pi,p)$ is admissible if and only if $X(T)\geq0$ and $\forall t\in[0,T]$,

\begin{equation}\label{xexpectation}
    X(t)+g(t)=\mathbb{E}^{\mathbb{Q},\psi}\Bigl[\int_t^Te^{-\int_t^s(r(u)+\mu(u))du}[c(s)+\mu(s)p(s)]ds +e^{-\int_t^T(r(u)+\mu(u))du}X(T)\mid \mathcal{F}_t\Bigl]\,.
\end{equation}
At time zero, this means that the strategies have to fulfill the following budget constraint:
\begin{equation}\label{budget}
    X(0)+g(0)=\mathbb{E}^{\mathbb{Q},\psi}\left[\int_0^Te^{-\int_0^t(r(u)+\mu(u))du}[c(t)+\mu(t)p(t)]dt +e^{-\int_0^T(r(u)+\mu(u))du}X(T)\right].
\end{equation}

Note that the condition (\ref{admissible}) allows the wealth to become negative, as long as it does not exceed in absolute value the actuarial value of future labor income $g(t)$ in \eqref{laborincome} so that it prevent the family from borrowing against the future labor income.\\

As in Kronborg and Steffensen \cite{steffensen}, the following remark is useful for the rest of the paper.\\

\noindent{\bf Remark.}
Define
\begin{equation}\label{martingalez}
    Y(t):=\int_0^te^{-\int_0^s(r(u)+\mu(u))du}[c(s)+\mu(s)p(s)-\ell(s)]ds+X(t)e^{-\int_0^t(r(u)+\mu(u))du}, \ \ t\in[0,T].
\end{equation}
By (\ref{wealthq}) we have that the Conditions (\ref{martingalew1}), \eqref{martingalew2} and (\ref{martingaleM}) are fulfilled if and only if $Y$ is a martingale under $\mathbb{Q}$. The natural interpretation is that, under $\mathbb{Q}$, the discounted wealth plus discounted pension contributions should be martingales. It is useful to note that if $Y$ is a martingale under $\mathbb{Q}$, the dynamics of $X$ can be represented in the following form:
\begin{eqnarray}\label{admissibility}
    dX(t) &=& [(r(t)+\mu(t))X(t)+\ell(t)-c(t)-\mu(t)p(t)]dt+\phi_1(t)dW_2^{\mathbb{Q},\psi}(t) \\ \nonumber
    &&  +\phi_2(t)dW_2^{\mathbb{Q},\psi}(t)+\varphi(t)d\tilde{N}^{\mathbb{Q},\psi}(t)\,,
\end{eqnarray}
for some $\mathcal{F}^W_t$-adapted processes $\phi_1\,, \phi_2(t)$ and $\mathcal{F}^N_t$-adapted process $\varphi$, satisfying $\phi(t),\varphi(t)\in L^2$, for any $ t\in[0,T]$, then under $\mathbb{Q}$, $Y$ is a martingale.

\section{The Unrestricted problem}

In this section, we solve our main optimization problem using the martingale duality method. Consider the concave, non-decreasing, upper semi-continuous and differentiable w.r.t. the second variable utility functions
$
U_k:[0,T]\times\mathbb{R}_+\rightarrow\mathbb{R}_+\,, \ \ k=1,2,3\,.
$
Define the strictly decreasing continuous inverse functions $ I_k:[0,T]\times\mathbb{R}_+\rightarrow\mathbb{R}_+\,, \ \ k=1,2,3\,,$ by
\begin{equation}\label{inverse}
I_k(t,x)=\left(\frac{\partial U_k(t,x)}{\partial x}\right)^{-1}\,, \ \ k=1,2,3.
\end{equation}
The Legendre-Frechel transform $\tilde{U}_k$ corresponding to the utility function $U_k$ is defined as follows: (see, Karatzas {\it et al} \cite{karatzas98})
\begin{equation}\label{legendre}
    \widetilde{U}_k(t,x):= \max_{y>0}[U_k(t,y)-yx]=U_k(t,I_k(t,x))-xI_k(t,x)\,, \ \ t\in[0,T]\,, \ 0<x<\infty\,.
\end{equation}

Let $\rho(t)$ be a deterministic function representing the policyholder's time preferences. The policyholder chooses his strategy $(c(t),\pi(t),p(t))$ in order to optimize the expected utility from consumption, legacy upon death and terminal pension. His strategy, therefore, fulfils the following:

\begin{eqnarray}\nonumber
    &&
J(x,c^*,\pi^*,p^*):=\sup_{(\pi,c,p)\in\mathcal{A}'}\mathbb{E}\Bigl[\int_0^{\tau\wedge T}
e^{-\int_0^s\rho(u)du}U_1(s,c(s))ds +e^{-\int_0^{\tau}\rho(u)du}U_2(\tau,p(\tau))\mathbf{1}_{\{\tau\leq T\}}  \\ \label{maximumutility} &&
 \ \ \ \ \ \ \  \ \ \ \ \ \ \ \ \ \ \ \ \ \ \ \ \ \ \ \ \ \ \ \ \ \ \   +e^{-\int_0^T\rho(u)du}U_3(X(T))\mathbf{1}_{\{\tau>T\}}\Bigl].
\end{eqnarray}
Here, $\mathbf{1}_A$ is an indicator function of a set $A$.
$\mathcal{A}'$ is the subset of the admissible strategies (feasible strategies) given by:
\begin{eqnarray}\nonumber
    \mathcal{A}' &:=& \Bigl\{(c,\pi,p)\in\mathcal{A}\mid \mathbb{E}\Bigl[\int_0^{\tau\wedge T}
e^{-\int_0^s\rho(u)du}\min(0,U_1(s,c(s)))ds \\ \nonumber
&& +e^{-\int_0^{\tau}\rho(u)du}\min(0,U_2(\tau,p(\tau)))\mathbf{1}_{\{\tau\leq T\}} \\ \label{feasibleset}
&& +e^{-\int_0^T\rho(u)du}\min(0,U_3(X(T)))\mathbf{1}_{\{\tau>T\}}\Bigl]>-\infty\Bigl\}.
\end{eqnarray}

The feasible strategy (\ref{feasibleset}) means that it is allowed to draw an infinite utility from the strategy $(\pi,c,p)\in\mathcal{A}'$, but only if the expectation over the negative parts of the utility function is finite. It is clear that for a positive utility function, the sets $\mathcal{A}$ and $\mathcal{A}'$ are equal ( see e.g., Kronborg and Steffensen \cite{steffensen}).
In order to solve the unrestricted control problem (\ref{maximumutility}), one can use the Hamilton-Jacobi-Bellman (HJB) equation (e.g. Mnif \cite{mnif}) or the combination of HJB equation with backward stochastic differential equation (BSDE) with jumps (Guambe and Kufakunesu \cite{guambe}). In this paper, we use the martingale approach applied in (Karatzas {\it et al.} \cite{karatzas}, Castaneda-Leyva and Hern\'andez-Hern\'andez \cite{hernandez}, Kronborg and Steffensen \cite{steffensen}). This is due to the restricted problem in the next section, where its terms are derived from the martingale method in the unrestricted problem.

Using (\ref{survival}) and (\ref{death}), we can rewrite the policyholder's optimization problem (\ref{maximumutility}) as:

\begin{eqnarray}\nonumber
 J(x,c^*,\pi^*,p^*) &=& \sup_{(c,\pi,p)\in\mathcal{A}'}\mathbb{E}\Bigl[\int_0^T
e^{-\int_0^s\rho(u)du}[\bar{F}(s)U_1(s,c(s))+f(s)U_2(s,p(s))]ds \\ \nonumber
   && \ \ \ \ \ \ \ \ \ \ \ \ \ \ \ \ \ \ \  +e^{-\int_0^T\rho(u)du}\bar{F}(T)U_3(X(T))\Bigl].
\end{eqnarray}
Hence,

\begin{eqnarray}\nonumber
 J(x,c^*,\pi^*,p^*) &=&  \sup_{(c,\pi,p)\in\mathcal{A}'}\mathbb{E}\Bigl[\int_0^T
e^{-\int_0^s(\rho(u)+\mu(u))du}[U_1(s,c(s)) +\mu(s)U_2(s,p(s))]ds \\ \label{optimization}
   && \ \ \ \ \ \ \ \ \ \ \ \ \ \ \ \ \ \ \ \  +e^{-\int_0^T(\rho(u)+\mu(u))du}U_3(X(T))\Bigl].
\end{eqnarray}

We now solve the main problem using the duality method. This approach allow us construct an auxiliary  market $\mathcal{M}_{\hat{\psi}}$ related to the original one, by searching over a family of martingale measures, the {\it inf-sup} martingale measure $\hat{\psi}$ and so the hedging portfolio process in the auxiliary market, satisfies the portfolio constraints in the original market $\mathcal{M}_{\psi}$ and replicates exactly the contingent claim almost surely. This approach has been applied under diffusion in a number of papers, see, for instance, He and Pearson \cite{he1991}, Karatzas and Shreve \cite{karatzas98}, Section 5.8, Casta\~neda-Leyva and  Hernandez-Hernandez \cite{hernandez}, Liang and Guo \cite{liang}.
Otherwise, one can complete the market by adding factitious risky assets in order to obtain a complete market, then apply the martingale approach to solve the optimal portfolio problem. For the market completion, we refer to Karatzas {\it et. al.} \cite{karatzas}, Runggaldier \cite{runggaldier}, Section 4., Corcuera {\it et. al.} \cite{corcuera2006}.

Thus, we define the associated dual functional $\Psi(\zeta,\psi)$ to the primal problem \eqref{optimization}, where $\zeta$ is the Lagrangian multiplier, by:

\begin{eqnarray*}
  \Psi(\hat{\zeta},\hat{\psi}) &:=& \sup_{\zeta>0;\psi>0}\Bigl\{\mathbb{E}\Bigl[\int_0^T
e^{-\int_0^s(\rho(u)+\mu(u))du}[U_1(s,c(s)) +\mu(s)U_2(s,p(s))]ds  \\
   && \ \ \ \ \ \ \  +e^{-\int_0^T(\rho(u)+\mu(u))du}U_3(X(T))\Bigl] +\zeta(x+g(0)) \\
   && -\zeta\left\{ \mathbb{E}^{\mathbb{Q},\psi}\left[\int_0^Te^{-\int_0^t(r(u)+\mu(u))du}[c(t)+\mu(t)p(t)]dt +e^{-\int_0^T(r(u)+\mu(u))du}X(T)\right]\right\}\Bigl\}.
\end{eqnarray*}
The dual problem that corresponds to the primal problem \eqref{optimization}, consists of
\begin{equation}\label{dualproblem}
\min_{\zeta>0,\psi>0}\Psi(\zeta,\psi).
\end{equation}
Note that (see Cuoco \cite{cuoco} or Karatzas {\it et al} \cite{karatzas98}, for more details)
\begin{eqnarray*}
&&\mathbb{E}^{\mathbb{Q},\psi}\left[\int_0^Te^{-\int_0^t(r(u)+\mu(u))du}[c(t)+\mu(t)p(t)]dt +e^{-\int_0^T(r(u)+\mu(u))du}X(T)\right] \\
&=& \mathbb{E}\left[\int_0^Te^{-\int_0^t(r(u)+\mu(u))du}\Gamma^{\psi}(t)[c(t)+\mu(t)p(t)]dt +e^{-\int_0^T(r(u)+\mu(u))du}\Gamma^{\psi}(T)X(T)\right],
\end{eqnarray*}
where we have defined the adjusted state price deflator $\Gamma$ by:
 \begin{eqnarray}\nonumber
   \Gamma^{\psi}(t) &:=& \Lambda(t)e^{\int_0^t(\rho(s)-r(s))ds} \\ \label{Gamma}
    &=& \exp\Bigl\{\int_0^t[\rho(s)-r(s)-\frac{1}{2}\theta^2(s,Z(s),\psi(s))-\frac{1}{2}\nu^2(s,Z(s),\psi(s))+(1-\psi(s))\lambda(s)]ds \\ \nonumber
&& +\int_0^t\nu(s,Z(s),\psi(s))dW_1(s)+\int_0^t\theta(s,Z(s),\psi(s))dW_2(s)+\int_0^t\ln(\psi(s))dN(s)\Bigl\}\,.
 \end{eqnarray}
This deflator can be written in the SDE form by:
\begin{eqnarray}\label{stateprice}
    d\Gamma^{\psi}(t)&=&\Gamma^{\psi}(t)\Bigl[(\rho(t)-r(t))dt+\nu(t,Z(s),\psi(t))dW_1(t)  \\ \nonumber
    && \ \ \ \ \ \ \ \ \ \ \ +\theta(t,R(t),\psi(t))dW_2(t)   +(\psi(t)-1)d\tilde{N}(t)\Bigl].
\end{eqnarray}

 Then, from the definition of the Legendre-Transform \eqref{legendre}, the dual functional $\Psi$ in \eqref{dualproblem} can be written as

\begin{eqnarray}\nonumber
  \Psi(\zeta,\psi) &:=& \mathbb{E}\Bigl[\int_0^T
e^{-\int_0^s(\rho(u)+\mu(u))du}[\widetilde{U}_1(s,c(s)) +\mu(s)\widetilde{U}_2(s,p(s))]ds   \\ \label{dualfunctional}
   && \ \ \ \ \ \ \ \ \ \ \ \ \   +e^{-\int_0^T(\rho(u)+\mu(u))du}\widetilde{U}_3(X(T))\Bigl] +\zeta(x+g(0))\,.
\end{eqnarray}

The following theorem shows, under suitable conditions the relationship between the primal problem \eqref{optimization} and the dual problem \eqref{dualproblem}.

\begin{theorem}
Suppose that $\hat{\psi}>0$ and $\hat{\zeta}>0$. The strategy $(c^*(t),p^*(t))\in\mathcal{A}'$ and $X^*(T)>0$ defined by
$$
c^*(t)=I_1(t,\hat{\zeta}\Gamma^{\hat{\psi}}(t)); \ \ p^*(t)=I_2(t,\hat{\zeta}\Gamma^{\hat{\psi}}(t)); \ \ X^*(T)=I_3(\hat{\zeta}\Gamma^{\hat{\psi}}(T)),
$$
such that \eqref{budget} is fulfilled, where $X^*(T)\in\mathcal{F}_T$ is measurable, is the optimal solution to the primal problem \eqref{optimization}, while $(\hat{\psi},\hat{\zeta})$ is the optimal solution to the dual problem \eqref{dualproblem}.
\end{theorem}

\begin{proof}
By the concavity of the utility functions $U_k, \ \ k=1,2,3$, (see Karatzas {\it et al} \cite{karatzas98}), we know that
$$
U_k(t,x)\leq U(t,I_k(t,x))-y(I_k(t,y)-x)\,.
$$
Then it can be easily shown that
\begin{equation}\label{inequality}
J(t,c(t),p(t),X(t))\leq\inf_{\zeta>0,\psi>0}\Psi(\zeta,\psi)\,.
\end{equation}
Hence, to finish the proof, we need to show that
$$
\inf_{\zeta>0,\psi>0}\Psi(\zeta,\psi)\geq J(t,c(t),p(t),X(t))\,.
$$
From \eqref{dualfunctional}, we know that
\begin{eqnarray*}
  && \inf_{\zeta>0,\psi>0}\Psi(\zeta,\psi)  \\
  &=& \inf_{\zeta>0,\psi>0}\Bigl\{\mathbb{E}\Bigl[\int_0^T
e^{-\int_0^s(\rho(u)+\mu(u))du}[\widetilde{U}_1(s,c(s)) +\mu(s)\widetilde{U}_2(s,p(s))]ds   \\
   && \ \ \ \ \ \ \ \ \ \ \ \ \   +e^{-\int_0^T(\rho(u)+\mu(u))du}\widetilde{U}_3(X(T))\Bigl] +\zeta(x+g(0))\Bigl\}  \\
   &\leq& \mathbb{E}\Bigl[\int_0^T
e^{-\int_0^s(\rho(u)+\mu(u))du}[\widetilde{U}_1(s,c(s)) +\mu(s)\widetilde{U}_2(s,p(s))]ds +e^{-\int_0^T(\rho(u)+\mu(u))du}\widetilde{U}_3(X(T))\Bigl] \\
&& +\hat{\zeta}(x+g(0)) \\
&=& \mathbb{E}\Bigl[\int_0^T
e^{-\int_0^s(\rho(u)+\mu(u))du}[U_1(s,c^*(s)) +\mu(s)U_2(s,p^*(s))]ds  \\
   && +e^{-\int_0^T(\rho(u)+\mu(u))du}U_3(X^*(T))\Bigl] -\hat{\zeta}\Bigl\{ \mathbb{E}\Bigl[\int_0^Te^{-\int_0^t(r(u)+\mu(u))du}\Gamma^{\hat{\psi}}(t)[c^*(t)+\mu(t)p^*(t)]dt  \\
   &&  +e^{-\int_0^T(r(u)+\mu(u))du}\Gamma^{\hat{\psi}}(T)X^*(T)\Bigl]\Bigl\} +\hat{\zeta}(x+g(0)) \\
   &=& \mathbb{E}\Bigl[\int_0^Te^{-\int_0^s(\rho(u)+\mu(u))du}[U_1(s,c^*(s)) +\mu(s)U_2(s,p^*(s))]ds \\
   && \ \ \ \ \ \ \ \ \ \ \ \ \ \ \ \ \ +e^{-\int_0^T(\rho(u)+\mu(u))du}U_3(X^*(T))\Bigl] \\
   &\leq& \sup_{(c,p,\pi)\in\mathcal{A}'}\Bigl\{\mathbb{E}\Bigl[\int_0^Te^{-\int_0^s(\rho(u)+\mu(u))du}[U_1(s,c(s))
    +\mu(s)U_2(s,p(s))]ds \\
   && \ \ \ \ \ \ \ \ \ \ \ \ \ \ \ \ \ \    +e^{-\int_0^T(\rho(u)+\mu(u))du}U_3(X(T))\Bigl]\Bigl\} \\
   &=& J(t,c(t),p(t),X(t))\,.
\end{eqnarray*}
Then, using \eqref{inequality}, we conclude the proof, i.e., $(c^*(t),p^*(t),X^*(T))$ is the optimal solution to the primal problem \eqref{optimization} and $(\hat{\psi},\hat{\zeta})$ is the optimal solution to the dual problem \eqref{dualproblem}.
\end{proof}

\noindent {\bf Remark.} Note that the optimal $(\hat{\psi},\hat{\zeta})$ is not necessarily unique, thus for  different choice of initial wealth, one might obtain different $\hat{\psi}$ and $\hat{\zeta}$.\\

\subsection{Results on the power utility case}
~~\\

In this section, we intend to derive the explicit solutions for the utility functions of the constant relative risk acersion type given by:

\begin{equation}\label{powerutility}
    U_1(t,x)=U_2(t,x)=U_3(t,x)=\left\{
           \begin{array}{ll}
             \frac{e^{-\kappa t}}{\delta}x^{\delta}, & \hbox{if} \ \ x>0, \\
             \lim_{x\rightarrow0}\frac{e^{-\kappa t}}{\delta}x^{\delta}, & \hbox{if} \ \  x=0, \\
             -\infty, & \hbox{if} \ \  x<0,
           \end{array}
         \right.
\end{equation}
for some $\delta\in(-\infty,1)\setminus\{0\}$ and $t\in[0,T]$. Thus the inverse function \eqref{inverse} is given by
\begin{equation}\label{inversepower}
I_k(t,x)=e^{-\frac{\kappa}{1-\delta}t}x^{-\frac{1}{1-\delta}}\,, \ \ \ k=1,2,3
\end{equation}
and the corresponding Legendre-Transform $\widetilde{U}_k$ by

$$
\widetilde{U}_k(t,x)=U_k(t,I_k(t,x))-xI_k(t,x)=\frac{1-\delta}{\delta}e^{-\frac{\kappa}{1-\delta}t}x^{-\frac{\delta}{1-\delta}}\,, \ \ k=1,2,3\,.
$$

We define a function $\mathcal{H}(\psi)$ by

\begin{eqnarray}\label{functionh}
\mathcal{H}(\psi) &:=& \mathbb{E}\Bigl[\int_0^Te^{-\int_0^t(\rho(u)+\mu(u)+\frac{\kappa}{1-\delta}u)du} [\Gamma^{\psi}(t)]^{-\frac{\delta}{1-\delta}}[1+\mu(t)]dt \\ \nonumber
&& \ \ \ \ \ \ \ \ \ \ \  +e^{-\int_0^T(\rho(u)+\mu(u)+\frac{\kappa}{1-\delta}u)du}[\Gamma^{\psi}(T)]^{-\frac{\delta}{1-\delta}}\Bigl]\,.
\end{eqnarray}
Then the dual functional \eqref{dualfunctional} is given by
\begin{equation}\label{dualp}
\Psi(\zeta,\psi)=\frac{1-\delta}{\delta}\zeta^{-\frac{\delta}{1-\delta}}\mathcal{H}(\psi)+\zeta(x+g(0))\,.
\end{equation}
Fixing $\psi>0$ and taking the minimum on \eqref{dualp}, we obtain the optimal $\hat{\zeta}$, given by

$$
\hat{\zeta}=\left[\frac{x+g(0)}{\mathcal{H}(\psi)}\right]^{\delta-1}\,.
$$
Inserting this optimal $\hat{\zeta}$ to the above equation, we obtain

$$
\Psi(\hat{\zeta},\psi)=\frac{1}{\delta}(x+g(0))^\delta \mathcal{H}^{1-\delta}(\psi)\,.
$$
Now, solving the dual problem \eqref{dualproblem} is equivalent to solving the following value function problem

\begin{equation}\label{newproblem}
V(t,Z(t))=\inf_{\psi>0}\mathcal{H}(\psi)\,, \ \ \ \delta>0
\end{equation}
or
\begin{equation}\label{newproblemsup}
V(t,Z(t))=\sup_{\psi>0}\mathcal{H}(\psi)\,, \ \ \ \delta<0\,.
\end{equation}

Note that from \eqref{stateprice} and the It\^o's formula (see, Oksendal and Sulem \cite{Oksendal}, Theorem 1.16), yields
\begin{eqnarray*}
  &&d[\Gamma^{\psi}(t)]^{-\frac{\delta}{1-\delta}} \\
  &=& [\Gamma^{\psi}(t)]^{-\frac{\delta}{1-\delta}}\Bigl\{\Bigl[-\frac{\delta}{1-\delta}(\rho(t)-r(t)) +\frac{\delta}{2(1-\delta)^2}(\nu^2(t,Z(t),\psi(t))+\theta^2(t,Z(t),\psi(t))) \\
&& +\bigl(\psi^{-\frac{\delta}{1-\delta}}(t)-1 +\frac{\delta}{1-\delta}(\psi(t)-1)\bigl)\lambda(t) \Bigl]dt
  -\frac{\delta}{1-\delta}\theta(t,Z(t),\psi(t))dW_2(t) \\
 && -\frac{\delta}{1-\delta}\nu(t,Z(t),\psi(t))dW_1(t) +\left( \psi^{-\frac{\delta}{1-\delta}}(t)-1\right)d\tilde{N}(t)\Bigl\}\,,
\end{eqnarray*}
which gives the following representation
\begin{eqnarray*}
  &&\mathbb{E}\left\{[\Gamma^{\psi}(t)]^{-\frac{\delta}{1-\delta}}\right\} \\
  &=& \mathbb{E}\Bigl[\exp\Bigl\{\int_0^t\Bigl[\frac{\delta}{1-\delta}(r(u)-\rho(u))+\frac{\delta}{2(1-\delta)^2}(\nu^2(u,Z(u),\psi(u)) +\theta^2(u,Z(u),\psi(u))) \\
   && \ \ \ \ \ \  +\bigl(\psi^{-\frac{\delta}{1-\delta}}(u)-1  +\frac{\delta}{1-\delta}(\psi(u)-1)\bigl)\lambda(u) \Bigl]du\Bigl\}\Bigl]\,; \ \ \ u\in[0,T]\,.
\end{eqnarray*}
Then the function $\mathcal{H}(\psi)$ can be written as
\begin{equation}\label{solutionh}
\mathcal{H}(\psi) = \mathbb{E}\Bigl[\int_0^Te^{-\int_0^t(\tilde{r}(u,Z(u),\psi(u))+\mu(u)+\frac{\kappa}{1-\delta}u)du} [1+\mu(t)]dt +e^{-\int_0^T(\tilde{r}(u,Z(u),\psi(u))+\mu(u)+\frac{\kappa}{1-\delta}u)du}\Bigl]\,,
\end{equation}
where
\begin{eqnarray}\nonumber
    \tilde{r}(t,Z(t),\psi(t))&=&-\frac{\delta}{1-\delta}r(t)+\frac{1}{1-\delta}\rho+\frac{\delta}{2(1-\delta)^2}(\nu^2(t,Z(t),\psi(t)) +\theta^2(t,Z(t),\psi(t))) \\ \label{rtilde}
    && +\left(\psi^{-\frac{\delta}{1-\delta}}(t)-1+\frac{\delta}{1-\delta}(\psi(t)-1)\right)\lambda(t)\,.
\end{eqnarray}

Proceeding as in \eqref{radom}, we define a new probability measure $\widetilde{\mathbb{Q}}$ equivalent to $\mathbb{P}$, by
$$
d\widetilde{\mathbb{Q}}=\Lambda^{-\frac{\delta}{1-\delta}}d\mathbb{P}\,.
$$
By this change of measure, the external economic factor \eqref{nontradable} can be written as

\begin{equation}\label{newnontradable}
    dZ(t)=\left[\eta(Z(t))-\frac{\delta}{1-\delta}\nu(t,Z(t),\psi(t))\right]dt +dW^{\widetilde{\mathbb{Q}},\psi}_1(t)\,.
\end{equation}
Now, the problem \eqref{newproblem} with $\mathcal{H}(\psi)$ given by \eqref{solutionh} can be solved using the dynamic programming approach. It is easy to see that the associated Hamilton-Jacobi-Bellman (HJB) equation satisfying $V(t,Z(t))$ is given by (see, Oksendal and Sulem \cite{Oksendal}, Theorem 3.1. for more details)

\begin{eqnarray*}
  1+\mu(t)+V_t(t,z)+\frac{1}{2}V_{zz}(t,z)+\Bigl[\frac{1}{1-\delta}(\delta r(t)-\rho-\delta\lambda(t)+\kappa t)-\lambda(t)+\mu(t)  && \\
  -\frac{\delta(r(t)-\alpha(t,z))^2}{2(1-\delta)^2[\beta^2(t,z)+\sigma^2(t,z)]}\Bigl]V(t,z)
  +\left[\eta(z)-\frac{\delta\beta(t,z)(r(t)-\alpha(t,z))}{(1-\delta)[\beta^2(t,z)+\sigma^2(t,z)]}\right]V_z(t,z) &&  \\
  -\inf_{\psi>0}\Bigl\{\left(\psi^{-\frac{\delta}{1-\delta}}+\frac{\delta}{1-\delta}\psi\right)\lambda(t) +\frac{\delta(\gamma^2(t,z)\lambda^2(t)\psi^2-2(r(t)-\alpha(t,z))\gamma(t,z)\lambda(t)\psi)}{2(1-\delta)^2 [\beta^2(t,z)+\sigma^2(t,z)]}V(t,z) && \\
   -\frac{\delta\beta(t,z)\gamma(t,z)\lambda(t)\psi}{(1-\delta)[\beta^2(t,z)+\sigma^2(t,z)]}V_z(t,z)\Bigl\}\,\,=\,\,0\,. &&
\end{eqnarray*}

We look for a candidate solution of the form $V(t,z)=\exp\{-h(t,z)\}$. Then we obtain the following semi-linear partial differential equation, whose existence and uniqueness of a smooth solution has been established, under the assumptions ${\bf (A_1)-(A_2)}$. (For more details we refer to Pham \cite{pham} or Mnif \cite{mnif}, Theorem 4.1.)

\begin{equation}\label{pde}
    -h_t(t,z)-\frac{1}{2}h_{zz}(t,z)+\frac{1}{2}h_z^2(t,z)+H(t,z,h,h_z,\psi)=0\,,
\end{equation}
where

\begin{eqnarray*}
  H(t,z,h,h_z,\psi) &=& \frac{1}{1-\delta}(\delta r(t)-\rho-\delta\lambda(t)+\kappa t)-\lambda(t)+\mu(t) \\
  && -\frac{\delta(r(t)-\alpha(t,z))^2}{2(1-\delta)^2[\beta^2(t,z)+\sigma^2(t,z)]} +(1+\mu(t))e^{h(t,z)} \\
   && +\Bigl[\eta(z)-\frac{\delta\beta(t,z)(r(t)-\alpha(t,z))}{(1-\delta)[\beta^2(t,z)+\sigma^2(t,z)]}\Bigl]h_z(t,z)  \\
  && -\inf_{\psi>0}\Bigl\{\Bigl(\psi^{-\frac{\delta}{1-\delta}}+\frac{\delta}{1-\delta}\psi\Bigl)\lambda(t)e^{h(t,z)} \\
  && +\frac{\delta(\gamma^2(t,z)\lambda^2(t)\psi^2-2(r(t)-\alpha(t,z))\gamma(t,z)\lambda(t)\psi)}{2(1-\delta)^2 [\beta^2(t,z)+\sigma^2(t,z)]} \\
  && -\frac{\delta\beta(t,z)\gamma(t,z)\lambda(t)\psi}{(1-\delta)[\beta^2(t,z)+\sigma^2(t,z)]}h_z(t,z)\Bigl\}\,.
\end{eqnarray*}

Now, suppose that there exists a unique solution $h(t,z)\in C^{1,2}([0,T)\times\mathbb{R})\times C^0([0,T]\times\mathbb{R})$ to the semi-linear equation \eqref{pde}. We define a function

\begin{eqnarray*}
  \mathcal{K}(t,z,\psi) &=& \Bigl(\psi^{-\frac{\delta}{1-\delta}}+\frac{\delta}{1-\delta}\psi\Bigl)\lambda(t)e^{h(t,z)} -\frac{\delta\beta(t,z)\gamma(t,z)\lambda(t)\psi}{(1-\delta)[\beta^2(t,z)+\sigma^2(t,z)]}h_z(t,z) \\
  &&  +\frac{\delta(\gamma^2(t,z)\lambda^2(t)\psi^2-2(r(t)-\alpha(t,z))\gamma(t,z)\lambda(t)\psi)}{2(1-\delta)^2 [\beta^2(t,z)+\sigma^2(t,z)]}\,.
\end{eqnarray*}
Note that from ${\bf (A_1)}$, $\mathcal{K}$ is continuous. Moreover, it is convex in a bounded $\psi>0$, i.e.,
$$
\frac{\partial^2\mathcal{K}}{\partial\psi^2}(t,z,\psi)=\frac{\delta\lambda(t)}{(1-\delta)^2}\Bigl[\psi^{\frac{\delta-2}{1-\delta}}e^{h(t,z)} +\frac{1}{\beta^2(t,z)+\sigma^2(t,z)}\Bigl]>0\,, \ \ \ \delta>0\,.
$$
So there exists
$$\hat{\psi}(t)\in\arg\min\mathcal{K}(t,z,\psi)\,.$$

The existence of the optimal measure $\hat{\psi}(t)$ for the problem \eqref{newproblemsup} can be solved similarly.\\

Since we obtained the optimal $\hat{\zeta}$ and $\hat{\psi}$, from Theorem 3.1 and \eqref{inversepower}, we obtain the following expressions

\begin{eqnarray}\label{optimalcd}
  c^*(t)=p^*(t) &=& \frac{X(t)+g(t)}{\mathcal{H}(t)}e^{-\frac{\kappa}{1-\delta}t}\,, \\ \label{optimalx}
  X^*(T) &=& \frac{X(t)+g(t)}{\mathcal{H}(t)}e^{-\frac{\kappa}{1-\delta}t}\left(\frac{\Gamma(T)}{\Gamma(t)}\right)^{-\frac{1}{1-\delta}}\,,
\end{eqnarray}
where

$$
\mathcal{H}(t)=\mathbb{E}\Bigl[\int_t^Te^{-\int_t^s(\tilde{r}(u,Z(u),\hat{\psi})+\mu(u)+\frac{\kappa}{1-\delta}(u-t)du} [1+\mu(s)]ds +e^{-\int_t^T(\tilde{r}(u,Z(u),\hat{\psi})+\mu(u)+\frac{\kappa}{1-\delta}(u-t))du}\Bigl]\,.
$$

From (\ref{stateprice}), by It\^o's formula we know that
\begin{eqnarray*}
  && \left(\frac{\Gamma(T)}{\Gamma(t)}\right)^{-\frac{1}{1-\delta}} \\
   &=& \exp\Bigl\{\frac{1}{1-\delta}\int_t^T\Bigl[ r(s)+\frac{1}{2}\nu^2(s,Z(s),\hat{\psi})+\frac{1}{2}\theta^2(s,Z(s),\hat{\psi})-\rho(s)+[\hat{\psi}-1-\ln\hat{\psi}]\lambda(s)\bigl]ds \\
   && -\frac{1}{1-\delta}\left[\int_t^T\nu(s,Z(s),\hat{\psi})dW_1(s) +\int_t^T\theta(s,Z(s),\hat{\psi})dW_2(s)+\int_t^T\ln\hat{\psi}d\tilde{N}(s)\right] \Bigl\}\,.
\end{eqnarray*}
Then we have:
\begin{eqnarray}\nonumber
    dX^*(t) &=& \mathcal{O}dt-\frac{1}{1-\delta}\nu(t)(X^*(t)+g(t))dW_1(t)-\frac{1}{1-\delta}\theta(t)(X^*(t)+g(t))dW_2(t) \\ \label{optimalwealth}
   && +\left(\hat{\psi}^{-\frac{1}{1-\delta}}(t) -1\right)(X^*(t)+g(t))dN(t),
\end{eqnarray}
where $\mathcal{O}:=\mathcal{O}(t,X^*(t),g(t))$. Comparing (\ref{optimalwealth}) with (\ref{wealth}), we obtain the optimal allocation:

\begin{equation}\label{optimalpi}
    \left\{
      \begin{array}{lll}
        \pi^*(t)\beta(t,Z(t)) =-\frac{1}{1-\delta}\nu(t,Z(t),\hat{\psi})(X^*(t)+g(t))\,,  \\
        \pi^*(t)\sigma(t,Z(t))=-\frac{1}{1-\delta}\theta(t,Z(t),\hat{\psi})(X^*(t)+g(t))\,,  \\
        \pi^*(t)\gamma(t,Z(t))=\left(\hat{\psi}^{-\frac{1}{1-\delta}}-1\right)(X^*(t)+g(t))\,.
      \end{array}
\right.
\end{equation}\\
Hence,
\begin{equation}\label{explicitpi}
    \pi^*(t)=\frac{\left(\hat{\psi}^{-\frac{1}{1-\delta}}-1\right) -\frac{1}{1-\delta}\nu(t,Z(t),\hat{\psi}) -\frac{1}{1-\delta}\theta(t,Z(t),\hat{\psi})}{\beta(t,Z(t))+\sigma(t,Z(t))+\gamma(t,Z(t))}(X^*(t)+g(t))\,.
\end{equation}

Inserting \eqref{optimalcd} and \eqref{optimalpi} into \eqref{wealthqq} we obtain the following geometric SDE which can be solved applying the It\^o formula (Oksendal and Sulem \cite{Oksendal}, Theorem 1.16)
\begin{eqnarray}\nonumber
  \frac{d(X^*(t)+g(t))}{X^*(t)+g(t)} &=& \Bigl[r(t)+\mu(t)-\frac{1+\mu(t)}{\mathcal{H}(t)} \Bigl]dt -\frac{1}{1-\delta}\nu(t,Z(t),\hat{\psi})dW_1^{\mathbb{Q},\hat{\psi}}(t) \\ \label{dynamicoptimalxg}
   &&  -\frac{1}{1-\delta}\theta(t,Z(t),\hat{\psi})dW_2^{\mathbb{Q},\hat{\psi}}(t) +\left(\hat{\psi}^{-\frac{1}{1-\delta}}(t)-1\right)d\widetilde{N}^{\mathbb{Q},\hat{\psi}}(t)\,.
\end{eqnarray}

We conclude this section, summarizing our results in the following Lemma:
\begin{lemma}
For the power utility functions \eqref{powerutility}, the optimal investment-consumption-insurance strategy $(c^*(t),\pi^*(t),p^*(t)), \ \forall t\in[0,T]$ is given by
$$
c^*(t)=p^*(t) = \frac{X^*(t)+g(t)}{\mathcal{H}(t)}e^{-\frac{\kappa}{1-\delta}t}
$$
and
\begin{equation*}
    \pi^*(t)=\frac{\left(\hat{\psi}^{-\frac{1}{1-\delta}}-1\right) -\frac{1}{1-\delta}\nu(t,Z(t),\hat{\psi}) -\frac{1}{1-\delta}\theta(t,Z(t),\hat{\psi})}{\beta(t,Z(t))+\sigma(t,Z(t))+\gamma(t,Z(t))}(X^*(t)+g(t))\,.
\end{equation*}

\end{lemma}

\section{The restricted control problem}
In this section, we solve the optimal investment, consumption and life insurance problem for the constrained control problem. We obtain an optimal strategy for the case of continuous constraints (American put options) by using a so-called {\it option based portfolio insurance (OBPI) strategy}. The OBPI method consists in taking a certain part of capital and invest in the optimal portfolio of the unconstrained problem and the remaining part insures the position with American put. We prove the admissibility and the optimality of the strategy. For more details see e.g., El Karoui {\it et. al.} \cite{elkaroui}, Kronborg and Steffensen \cite{steffensen}.

Consider the following problem
\begin{eqnarray}\nonumber
  &&  \sup_{(c,\pi,D)\in\mathcal{A}'}\mathbb{E}\Bigl[\int_0^T
e^{-\int_0^s(\rho(u)+\mu(u))du}[U(c(s)) +\mu(s)U(p(s))]ds \\ \label{restricted}
   && +e^{-\int_0^T(\rho(u)+\mu(u))du}U(X(T))\Bigl]\,,
\end{eqnarray}
under the capital guarantee restriction
\begin{equation}\label{capitalrestriction}
    X(t)\geq k(t,D(t)), \ \forall t\in[0,T],
\end{equation}
where
\begin{equation*}
    D(t):=\int_0^th(s,X(s))ds,
\end{equation*}
where $k$ and $h$ are deterministic functions of time. The guarantee \eqref{capitalrestriction} is covered by
\begin{equation}\label{guarantee0}
    k(t,d)=0
\end{equation}
and
\begin{equation}\label{guaranteek}
    k(t,d)=x_0e^{\int_0^t(r^{(g)}(s)+\mu(s))ds}+de^{\int_0^t(r^{(g)}(s)+\mu(s))ds},
\end{equation}
with
\begin{equation*}
    h(s,x)=e^{-\int_0^s(r^{(g)}(u)+\mu(u))du}[\ell(s)-c(s,x)-\mu(s)p(s,x)],
\end{equation*}
where $r^{(g)}\leq r$ is the minimum rate of return guarantee excess of the objective mortality $\mu$. Then
\begin{equation}\label{guaranteekk}
    k(t,z)=x_0e^{\int_0^t(r^{(g)}(s)-\mu(s))ds}+\int_0^te^{\int_s^t(r^{(g)}(u)+\mu(u))ds}[\ell(s)-c(s)-\mu(s)p(s)]ds.
\end{equation}

We still denote by $X^*,c^*,\pi^*$ and $p^*$ the optimal wealth, optimal consumption, investment and life insurance for the unrestricted problem (\ref{maximumutility}), respectively. The optimal wealth for the unrestricted problem $Y^*(t):=X^*(t)+g(t)$ has the dynamics
\begin{eqnarray}\label{wealtha}
  dY^*(t) &=&  Y^*(t)\Bigl\{\left[r(t)+\mu(t)-\frac{1+\mu(t)}{\mathcal{H}(t)}\right]dt -\frac{1}{1-\delta}\nu(t,Z(t),\hat{\psi})dW_1^{\mathbb{Q},\hat{\psi}}(t) \\ \nonumber
&& -\frac{1}{1-\delta}\theta(t,Z(t),\hat{\psi})dW_2^{\mathbb{Q},\hat{\psi}}(t) +\left(\hat{\psi}^{-\frac{1}{1-\delta}}(t)-1\right)d\widetilde{N}^{\mathbb{Q},\hat{\psi}}(t)\Bigl\}\,,
\end{eqnarray}
$\forall t\in[0,T], \ \ Y^*(0)=y_0,$ where $y_0:=x_0+g(0)$. Under $\mathbb{Q}$, the economic factor $Z$ is given by

\begin{equation}\label{newnontradableq}
    dZ(t)=\left[\eta(Z(t))+\nu(t,Z(t),\hat{\psi}(t))\right]dt +dW^{\mathbb{Q},\hat{\psi}}_1(t)\,.
\end{equation}

 Let $P_{y,d}^a(t,T,k+g)$ denote the time-$t$ value of an American put option with strike price $k(s,D(s))+g(s)$, $\forall s\in[t,T]$, where $D(t)=d$ and maturity $T$ written on a portfolio $Y$, where $Y(s), \ \ s\in[t,T]$ is the solution to (\ref{wealtha}), with $Y(t)=y$. By definition, the price of such put option is given by

\begin{eqnarray*}
  P_{y,d}^a(t,T,k+g) &:=& \sup_{\tau_s\in\mathcal{T}_{t,T}}\mathbb{E}^{\mathbb{Q}}\Bigl[e^{-\int_t^{\tau_s}(r(u)+\mu(u))du}[k(\tau_s,D(\tau_s)) +g(\tau_s) \\
   && -Y(\tau_s)]^+ \Bigl|Y(t)=y,D(t)=d\Bigl],
\end{eqnarray*}
where $\mathcal{T}_{t,T}$ is the set of stopping times $\tau_s\in[t,T]$.

As in Kronborg and Steffensen \cite{steffensen}, we introduce the American put option-based portfolio insurance
\begin{equation}\label{obpi}
    \hat{X}^{(\varrho)}(t):=\varrho(t,D(t))Y^*(t)+P_{\varrho Y^*\,,\,D}^a(t,T,k+g)-g(t), \ \ t\in[0,T],
\end{equation}
for $\varrho\in(0,1)$ defined by
\begin{equation}\label{varrho}
    \varrho(t,D(t))=\varrho_0\vee\sup_{s\leq t}\left(\frac{b(s,D(s))}{Y^*(s)}\right),
\end{equation}
where $b(t,D(t))$ is the exercise boundary of the American put option given by
\begin{equation}\label{boundary}
    b(t,d):=\sup\left\{y:P_{y,d}^a(t,T,k+g)=(k(t,d)+g(t)-y)^+\right\}\,
\end{equation}
and $\varrho_0:=\varrho(0,D(0))$ is determined by the budget constraint
\begin{equation}\label{obpi0}
    \varrho(0,D(0))Y^*(0)+P_{\varrho Y^*\,,\,D}^a(0,T,k+g)-g(0)=x_0.
\end{equation}

By definition of an American put option,  $P_{\varrho Y^*\,,D\,}^a(t,T,k+g)\geq(k(t,d)+g(t)-\varrho(t,D(t)) Y^*(t))^+$, $\forall t\in[0,T]$. Hence

\begin{eqnarray*}
  \widehat{X}^{(\varrho)}(t) &:=& \varrho(t,D(t))Y^*(t)+P_{\varrho Y^*,D}^a(t,T,k+g)-g(t) \\
   &\geq& \varrho(t,D(t))Y^*(t)+(k(t,d)+g(t)-\varrho(t,D(t)) Y^*(t))^+-g(t) \\
   &\geq& k(t,d), \forall t\in[0,T],
\end{eqnarray*}
i.e., $\widehat{X}^{(\varrho)}$ fulfils the American capital guarantee.

Under the optimal martingale measure $\hat{\psi}$, we recall some basic properties of American put options in a Black-Scholes market (Karatzas and Shreve, \cite{karatzas98})

\begin{equation*}
      \begin{array}{ll}
        P_{y,d}^a(t,T,k+g)=k(t,d)+g(t)-y, & \forall (t,y,d)\in\mathcal{C}^c \\
        \frac{\partial}{\partial y}P_{y,d}^a(t,T,k+g)=-1, & \forall (t,y,d)\in\mathcal{C}^c \\
        AP_{y,d}^a(t,T,k+g)=(r(t)+\mu(t))P_{y,d}^a(t,T,k+g), & \forall (t,y,d)\in\mathcal{C},
      \end{array}
\end{equation*}
where from (\ref{wealtha}), the generator operator $A$ is given by (see e.g. Oksendal and Sulem \cite{Oksendal}, Li {\it et. al} \cite{li})
\begin{eqnarray*}
  (A\phi)(y,z) &=& \frac{\partial \phi}{\partial t}+ \left(r(t)+\mu(t)- \frac{1+\mu(t)}{\mathcal{H}(t)}\right)y\frac{\partial\phi}{\partial y} +\left(\eta(z)+\nu(t,z,\hat{\psi})\right)\frac{\partial\phi}{\partial z} +\frac{1}{2}\frac{\partial^2\phi}{\partial z^2} \\
 && +\frac{1}{2(1-\delta)^2}\left[\nu^2(t,Z(t),\hat{\psi}) + \theta^2(t,Z(t),\hat{\psi})\right]y^2\frac{\partial^2\phi}{\partial y^2}  -\frac{1}{1-\delta}\nu(t,z,\hat{\psi})\frac{\partial^2\phi}{\partial y\partial z} \\
   && +\left[\phi(t,y\hat{\psi}^{-\frac{1}{1-\delta}},z)-\phi(t,y,z) -y\left(\hat{\psi}^{-\frac{1}{1-\delta}}-1\right)\frac{\partial \phi}{\partial y}\right]\lambda(t)
\end{eqnarray*}
and
\begin{equation*}
    \mathcal{C}:=\{(t,y,d):P_{y,d}^a(t,T,k+g)>(k(t,d)+g(t)-y)^+\}
\end{equation*}
defines the continuation region. $\mathcal{C}^c$ is the stopping region, that is the complementary of the continuation region $\mathcal{C}$. From the exercise boundary given in (\ref{boundary}), we can write the continuation region by
\begin{equation*}
    \mathcal{C}=\{(t,y,d):y>b(t,d)\}.
\end{equation*}
Define a function $H$ by
\begin{equation*}
    H(t,y,d):=y+P_{y,d}^a(t,T,k+g)-g(t),
\end{equation*}
then we have
$$
\widehat{X}^{(\varrho)}(t)=H(t,\varrho(t,D(t)) Y^*(t),D(t)).
$$
From the properties of $P_{y,d}^a(t,T,k+g)$, we deduce that
\begin{eqnarray}\nonumber
  H(t,y,d) &=& k(t,d), \ \ \forall (t,y,d)\in\mathcal{C}^c, \\ \label{partialh}
  \frac{\partial}{\partial y}H(t,y,d) &=& 0, \ \ \forall (t,y,d)\in\mathcal{C}^c \\ \label{partialk}
 AH(t,y,d) &=& \frac{\partial}{\partial t}k(t,d)+h(t,d)\frac{\partial}{\partial d}k(t,d)  \ \ \forall (t,y,d)\in\mathcal{C}^c, \\ \nonumber
AH(t,y,d) &=& (r(t)+\mu(t))P_{y,d}^a(t,T,k+g)+\ell(t)-(r(t)+\mu(t))g(t) \\ \nonumber
&&+\left(r(t)+\mu(t)-\frac{1+\mu(t)}{\mathcal{H}(t)}\right)y \\ \nonumber
&&+\left(P_{y\hat{\psi}^{-\frac{1}{1-\delta}},d}^a(t,T,k+g)-P_{y,d}^a(t,T,k+g)\right)\lambda(t) \\ \label{operatorA}
&=& (r(t)+\mu(t))H(t,y,d)+\ell(t)-\frac{1+\mu(t)}{\mathcal{H}(t)}y+\Bigl[H(t,y\hat{\psi}^{-\frac{1}{1-\delta}},d) \\ \nonumber
&& -H(t,y,d)-y\left(\hat{\psi}^{-\frac{1}{1-\delta}}(t)-1\right)\Bigl]\lambda(t), \ \ \forall (t,y,d)\in\mathcal{C}.
\end{eqnarray}
\begin{proposition}
Consider the strategy $(\varrho c^*,\varrho\pi^*,\varrho p^*)$, where $\varrho$ is defined by (\ref{varrho}). Then, the strategy $(\varrho c^*,\varrho\pi^*,\varrho p^*)$, where $\varrho$  is admissible.

\end{proposition}
\begin{proof}
For $\varrho$ constant and linearity of $Y^*(t),\forall t\in[0,T]$, we have that $\varrho(t,D(t)) Y^*(t)$ and $Y^*(t)$ have the same dynamics. Then, using It\^o's formula, (\ref{partialk})-(\ref{operatorA}), $(c^*(t),p^*(t))$ in Theorem 3.1 and the fact that $\varrho$ increases only at the boundary, we obtain (here $\frac{\partial}{\partial y}$ means differentiating with respect to the second variable)
\begin{eqnarray*}
  &&dH(t,\varrho(t,D(t))Y^*(t), D(t)) \\
&=& [dH(t,\varrho Y^*(t),D(t))]+Y^*(t)\frac{\partial}{\partial y}H(t,\varrho(t,D(t))Y^*(t),D(t))d\varrho(t,D(t)) \\
   &=& AH(t,\varrho Y^*(t))dt-\frac{1}{1-\delta}\nu(t,Z(t),\hat{\psi})\varrho Y^*(t) \frac{\partial}{\partial y}H(t,\varrho Y^*(t),D(t))dW_1^{\mathbb{Q},\hat{\psi}}(t) \\
   && -\frac{1}{1-\delta}\theta(t,Z(t),\hat{\psi})\varrho Y^*(t) \frac{\partial}{\partial y}H(t,\varrho Y^*(t),D(t))dW_2^{\mathbb{Q},\hat{\psi}}(t) \\
   &&  +\left[H(t,\varrho Y^*(t)\hat{\psi}^{-\frac{1}{1-\delta}}(t),D(t))-H(t,\varrho Y^*(t),D(t))\right]d\tilde{N}^{\mathbb{Q},\hat{\psi}}(t) \\
&& + Y^*(t)\frac{\partial}{\partial y}H(t,\varrho(t,D(t))Y^*(t),D(t))d\varrho(t,D(t)) \\
&=& \Bigl\{(r(t)+\mu(t))H(t,\varrho Y^*(t),D(t))+\ell(t)-\varrho c^*(t)-\varrho\mu(t)p^*(t) \\
&&+ \Bigl[H(t,\varrho Y^*(t)\hat{\psi}^{-\frac{1}{1-\delta}}(t),D(t))-H(t,\varrho Y^*(t),D(t)) \\
&&-\varrho Y^*(t)\left(\hat{\psi}^{-\frac{1}{1-\delta}}(t)-1\right)\Bigl]\lambda(t)\Bigl\} \mathbf{1}_{(\varrho(t,D(t))Y^*(t)>b(t,D(t)))}dt \\
&& \left[\frac{\partial}{\partial t}k(t,D(t))+h(t,D(t))\frac{\partial}{\partial d}k(t,D(t))\right] \mathbf{1}_{(\varrho(t,D(t))Y^*(t)\leq b(t,D(t)))}dt \\
&& + Y^*(t)\frac{\partial}{\partial y}H(t,\varrho(t,D(t))Y^*(t),D(t))\mathbf{1}_{(\varrho(t,D(t))Y^*(t)=b(t,D(t)))} d\varrho(t,D(t)) \\
&& -\frac{1}{1-\delta}\nu(t,Z(t),\hat{\psi})\varrho Y^*(t) \frac{\partial}{\partial y}H(t,\varrho Y^*(t),D(t))dW_1^{\mathbb{Q},\hat{\psi}}(t) \\
&& -\frac{1}{1-\delta}\theta(t,Z(t),\hat{\psi})\varrho Y^*(t) \frac{\partial}{\partial y}H(t,\varrho Y^*(t),D(t))dW_2^{\mathbb{Q},\hat{\psi}}(t) \\
&&+\left[H(t,\varrho Y^*(t)\hat{\psi}^{-\frac{1}{1-\delta}}(t),D(t))-H(t,\varrho Y^*(t),D(t))\right]d\tilde{N}^{\mathbb{Q},\hat{\psi}}(t).
\end{eqnarray*}
From (\ref{partialh}) we know that $\frac{\partial}{\partial y}H(t,\varrho(t,D(t))Y^*(t),D(t))=0$ on the set \\$\{(t,\omega):\varrho(t,D(t))Y^*(t) =b(t,D(t))\}$, then
\begin{eqnarray*}
   && dH(t,\varrho(t,D(t))Y^*(t), D(t)) \\
   &=& \Bigl\{(r(t)+\mu(t))H(t,\varrho Y^*(t),D(t))+\ell(t)-\varrho c^*(t)-\varrho\mu(t)p^*(t) \\
   &&  + \Bigl[H(t,\varrho Y^*(t)\hat{\psi}^{-\frac{1}{1-\delta}}(t),D(t))-H(t,\varrho Y^*(t),D(t))-\varrho Y^*(t)\left(\hat{\psi}^{-\frac{1}{1-\delta}}(t)-1\right)\Bigl]\lambda(t)\Bigl\}dt \\
&& + \Bigl[\frac{\partial}{\partial t}k(t,D(t))+h(t,D(t))\frac{\partial}{\partial d}k(t,D(t))-[(r(t)+\mu(t))k(t,D(t))+\ell(t) \\
&&-\varrho(t,D(t)) c^*(t)-\varrho(t,D(t))\mu(t)p^*(t)]\Bigl]\mathbf{1}_{(\varrho(t,D(t))Y^*(t)\leq b(t,D(t)))}dt \\
&& -\frac{1}{1-\delta}\nu(t,Z(t),\hat{\psi})\varrho Y^*(t) \frac{\partial}{\partial y}H(t,\varrho Y^*(t),D(t))dW_1^{\mathbb{Q},\hat{\psi}}(t) \\
&& -\frac{1}{1-\delta}\theta(t,Z(t),\hat{\psi})\varrho Y^*(t) \frac{\partial}{\partial y}H(t,\varrho Y^*(t),D(t))dW_2^{\mathbb{Q},\hat{\psi}}(t) \\
&&+\left[H(t,\varrho Y^*(t)\hat{\psi}^{-\frac{1}{1-\delta}}(t),D(t))-H(t,\varrho Y^*(t),D(t))\right]d\tilde{N}^{\mathbb{Q},\hat{\psi}}(t).
\end{eqnarray*}
Hence, since $\{(t,\omega):\varrho(t,D(t))Y^*(t) \leq b(t,D(t))\}=\left\{(t,\omega):\varrho(t,D(t))=\frac{b(t,D(t))}{Y^*(t)}\right\}$ has a zero $dt\otimes d\mathbb{P}$-measure, we conclude that
\begin{eqnarray*}
   && dH(t,\varrho(t,D(t))Y^*(t), D(t)) \\
   &=& \Bigl\{(r(t)+\mu(t))H(t,\varrho Y^*(t),D(t))+\ell(t)-\varrho c^*(t)-\varrho\mu(t)p^*(t) \\
   &&  + \Bigl[H(t,\varrho Y^*(t)\psi^{-\frac{1}{1-\delta}}(t),D(t))-H(t,\varrho Y^*(t),D(t))-\varrho Y^*(t)\left(\psi^{-\frac{1}{1-\delta}}(t)-1\right)\Bigl]\lambda(t)\Bigl\}dt \\
&& -\frac{1}{1-\delta}\nu(t,Z(t),\hat{\psi})\varrho Y^*(t) \frac{\partial}{\partial y}H(t,\varrho Y^*(t),D(t))dW_1^{\mathbb{Q},\hat{\psi}}(t) \\
&& -\frac{1}{1-\delta}\theta(t,Z(t),\hat{\psi})\varrho Y^*(t) \frac{\partial}{\partial y}H(t,\varrho Y^*(t),D(t))dW_2^{\mathbb{Q},\hat{\psi}}(t) \\
&&+\left[H(t,\varrho Y^*(t)\hat{\psi}^{-\frac{1}{1-\delta}}(t),D(t))-H(t,\varrho Y^*(t),D(t))\right]d\tilde{N}^{\mathbb{Q},\hat{\psi}}(t),
\end{eqnarray*}
i.e. by (\ref{admissibility}), the strategy $(\varrho c^*, \varrho\pi^*,\varrho p^*)$ is admissible.
\end{proof}

We then state the main result of this section.
\begin{theorem}
Consider the strategy $(\widehat{c},\widehat{\pi},\widehat{p})$, $\forall t\in[0,T]$ given by
\begin{eqnarray}
  \widehat{c} &=& \frac{\varrho(t,D(t))Y^*(t)}{\mathcal{H}(t)}= \varrho(t,D(t))c^*(t), \\
  \widehat{\pi} &=& \varrho(t,D(t))\pi^*(t), \\
  \widehat{p} &=& \frac{\varrho(t,D(t))Y^*(t)}{\mathcal{H}(t)}=\varrho(t,D(t))p^*(t),
\end{eqnarray}
where the strategy $(c^*,\pi^*,p^*)$ is defined in Lemma 3.2. Combined with a position in an American put option written on the portfolio $(\varrho(s,D(s))Y^*(s))$ with strike price $k(s,D(s))+g(s)$, $\forall s\in[t,T]$ and maturity $T$, where $\varrho(s,D(s)),s\in[t,T]$ is a function defined by (\ref{varrho}). Then, the strategy is optimal for the American capital guarantee control problem given by (\ref{restricted})-(\ref{capitalrestriction}).
\end{theorem}

\begin{proof}
Similar to that in \cite{steffensen}, Theorem 4.1.

\end{proof}

\section{Conclusion}
The paper focused on an optimal investment-consumption insurance with capital constraints, specifically the American capital guarantee. We solved our problem in the jump diffusion framework using the martingale approach. An explicit solution was obtained in the power utility case.

\subsection*{Acknowledgment}
We would like to express our deep gratitude to the University of Pretoria and the MCTESTP Mozambique for their support. We also wish to thank the anonymous referee for his comments and critics during the reviewing process, they really helped a lot for the improvement of this paper.

\end{document}